\documentclass[%
reprint,
nofootinbib,
 amsmath,amssymb,
 aps,
rmp,
]{revtex4-2}

\makeatletter 
\renewcommand\onecolumngrid{
\do@columngrid{one}{\@ne}%
\def\set@footnotewidth{\onecolumngrid}
\def\footnoterule{\kern-6pt\hrule width 1.5in\kern6pt}%
}
\renewcommand\twocolumngrid{
        \def\footnoterule{
        \dimen@\skip\footins\divide\dimen@\thr@@
        \kern-\dimen@\hrule width.5in\kern\dimen@}
        \do@columngrid{mlt}{\tw@}
}%
\makeatother

\usepackage[braket,qm]{qcircuit}
\usepackage{amsmath}
\usepackage{amssymb}
\usepackage{geometry}
\usepackage{color}
\usepackage{amsthm}
\usepackage{graphicx}
\usepackage{mathtools}
\usepackage{wasysym}
\usepackage{color, colortbl}
\usepackage{empheq}
\usepackage{framed}
\usepackage{dirtree}
\usepackage{scrextend}
\usepackage{mathpartir}
\usepackage{verbatim}
\usepackage{stmaryrd}

\newtheorem{definition}{Definition}
\newtheorem{theorem}{Theorem}

\newtheorem{nontheorem}[theorem]{Non-theorem}

\usepackage{graphicx}
\usepackage{dcolumn}
\usepackage{bm}

\newcommand{\norm}[1]{\left\|#1\right\|}

\newcommand{\set}[1]{\left\lbrace #1 \right\rbrace}

\newcommand{\ketbra}[2]{\left\vert#1\right\rangle\left\langle#2\right\vert}

\newcommand{\denotation}[1]{\left\llbracket#1\right\rrbracket}

\begin{document}

\title{Reasoning across spacelike surfaces in the Frauchiger-Renner thought experiment}
\author{Jeffrey M. Epstein}
\email{jeffrey.m.epstein@gmail.com}

\date{\today}

\begin{abstract}
The Frauchiger-Renner argument purports to show that the standard framework of quantum mechanics yields a contradiction when used to reason about systems containing agents who are themselves using quantum mechanics to perform deductions. This has been framed as an obstacle to taking quantum mechanics to be a complete theory. I formalize the argument in two closely related ways and elucidate the source of the paradox, clarifying the flaw in the original argument.
\end{abstract}

\maketitle
\onecolumngrid

\section{Introduction}
	The standard presentation of quantum mechanics distinguishes between two classes of basic entity: systems and observers. Systems are taken to evolve unitarily (or as part of a collection of systems jointly evolving unitarily) until measured by an observer, at which point something discontinuous may occur. Observers are simply those things capable of making observations. This dichotomy is left unexplained in the standard presentation. In fact, it goes against the basic intuitions most of us have about the universe. Assuming some form of materialism and reductionism, we model observers as being ultimately composed of systems. As quantum mechanics tells us exactly how to deal with composites of systems, it remains unclear how a distinction between systems and observers could arise. Without this distinction, however, it is unclear how the experience of definite measurement outcomes could arise. This tension is a basic motivator for a large body of interesting work in quantum foundations and the philosophy of quantum mechanics, and indeed is something that most people who have learned quantum mechanics have probably wondered about.
	
	There is a long history of thought experiments designed to address this problem. Schr\"odinger's cat (as analyzed by Schr\"odinger himself as well as many others) is an early attempt at precising this confusion, and Wigner continued in this tradition with the so-called Wigner's Friend thought experiment. In the past few years, there has been extensive discussion centered on an ingenious thought experiment of Frauchiger and Renner, first proposed by \cite{Frauchiger2018} and explicated somewhat more accessibly by \cite{Nurgalieva2020} and \cite{Bub2021}. They frame their protocol in terms of an apparent paradox arising when systems capable of using quantum mechanics to perform deductions are themselves treated quantum mechanically.
	
	In this essay, I will attempt to add to this conversation by providing a formalization of the Frauchiger-Renner paradox.  I argue that the proper way to understand the result is as a consequence of the improper chaining of inferences along different spacelike surfaces, in a way that I will make precise. In my view, this demonstrates the core error in one of the two apparently contradictory deductive procedures presented by Frauchiger and Renner to analyze their protocol, and removes any confusion arising from the less formal analysis of the original presentation. I begin in Section \ref{sec:protocol} by describing the protocol involved in the thought experiment. In \ref{sec:FRargs}, I summarize the arguments of Frauchiger and Renner that produce the apparent contradiction. In \ref{sec:analysis}, I introduce a framework that allows more formal reasoning about the protocol, and reframe the protocol as a counterexample to two apparently reasonable theorems asserting the validity of chained inferences within quantum mechanics. The first of these makes contact with the Everettian relative state framework, and the second with the program of quantum logic. In both cases, the corresponding theorem holds if we restrict to a single spacelike surface, so that we may identify the source of the paradox in the use of multiple surfaces in a single deduction.
	
	\section{Frauchiger-Renner Protocol}
	\label{sec:protocol}
	In this section I recapitulate the thought experiment (almost) as presented by Frauchiger and Renner. The protocol involves four agents\footnote{I leave out the systems $\bar{A}$ and $\bar{B}$ corresponding to Alice's and Bob's labs, or rather include these in $A$ and $B$. It does not seem to me that separating these out serves any purpose in the argument, and after all $A$ and $B$ themselves are already enormously complicated systems, supposed to include all of the sensory and cognitive equipment used by the agents. When Ursula measures $A$ and $R$, she is meant to be performing a projective measurement on this entire system, i.e. on all physical degrees of freedom with which $S$ has become entangled after the second step of the protocol.}, Alice, Bob, Ursula, and Wigner, and two qubits, R and S, prepared in the states $\sqrt{1/3}\ket{0}+\sqrt{2/3}\ket{1}$ and $\ket{0}$, respectively. It consists of the following steps:
	\begin{enumerate}
		\item Alice measures R in the computational basis.
		\item If she obtains the outcome 0, she does nothing. If she obtains the outcome 1, she applies a Hadamard gate to S.
		\item Bob measures S in the computational basis.
		\item Ursula measures the joint system composed of Alice and R in a basis containing the states 
		\begin{align}
			\ket{\pm}_{RA}=\left(\ket{0}_R\ket{0}_A\pm\ket{1}_R\ket{1}_A\right)/\sqrt{2},
		\end{align}
		where $\ket{0/1}_A$ is the state describing Alice having obtained the outcome $0/1$ upon performing her measurement\footnote{This state may also be taken to represent the entire state of Alice's brain after obtaining the measurement outcome \textit{and making deductions based on this outcome}. The actual process of making such deductions could be separated out from the operation $M_A$ and independently described by a unitary operator acting on Alice alone, or at a more granular level acting on two (or more) systems representing Alice's observation of the outcome $r$ and the results of her reasoning based on this outcome. This further granularity does not seem to me to add anything to the argument.}.
		\item Wigner measures the joint system composed of Bob and S in a basis containing the states $\ket{\pm}_{SB}$, defined analogously.
	\end{enumerate}
	We assume that the four agents may all be described as initialized in some state $\ket{\bot}$ corresponding to their situation prior to observing measurement outcomes and making deductions\footnote{A natural concern arises regarding the definedness of such states. After all, any description of these states in terms of basic quantities like spin, momentum, and position will be enormously complex. This is not a problem here -- all that is required is the claim that there are some quantum states corresponding to the states of affairs prior to measurement, and following each measurement outcome, and that states corresponding to classically distinct states of affairs are orthogonal. Put another way, these states are as well defined as the spin up and down states of an electron -- we don’t ever really say anything about these states beyond naming them. If we take the states to be complete descriptions of the state of affairs, then that’s all there is to it. Of course the question of phenomenological aspects of a state corresponding to an agent’s knowledge or beliefs is entirely mysterious, but this is true in a classical theory of conscious agents as well.}. The Hilbert space corresponding to each of the agents is at least large enough to include this state as well as the states $\ket{0}$ and $\ket{1}$. This protocol may be illustrated graphically in terms of the circuit diagram of Fig. \ref{Fig:causal_structure}. For the moment only the wires and boxes of the diagram are meaningful.

	\section{Paradox according to Frauchiger and Renner}
	\label{sec:FRargs}
	The purpose of the protocol designed by Frauchiger and Renner is to demonstrate an apparent paradox that arises when using quantum mechanics to reason about involving systems considered both (1) as agents capable of performing measurements of other systems\footnote{In the sense of updating their state based on the state of a system to be observed, \textit{not} necessarily in the sense of causing an objective collapse of the wavefunction.} and performing deductions using the framework of quantum mechanics, and (2) as quantum systems in their own right, which may themselves be measured in arbitrary bases by other agents. To this end, consider the situation in which Ursula observes the outcome $\ket{-}$ upon performing her measurement of the joint system consisting of Alice and the qubit R. We can make two different arguments to establish the possibility or impossibility of Wigner measuring the same outcome in this situation\footnote{In fact, in the original protocol, Ursula is supposed to ``classically communicate" the outcome of her reasoning to Wigner, who will then make two deductions about the value of his measurement, one based on the principle that he is correct about his own observations, and the other based on the principle that it is valid to adopt another agent's valid deductions. It does not seem to me that this further step adds anything to the paradox, and in fact I believe it obscures what I will later argue is the proper way to understand the origin of the apparent contradiction by treating the two Arguments I will present here as being of different types, whereas I view both as having, at the end of the day, the same form. Nevertheless, Frauchiger and Renner's original inclusion of this final communication from Ursula to Wigner does have the benefit of in some way internalizing the apparent paradox to the system by localizing it to Wigner, who allegedly must believe a contradiction with nonzero probability.}. We shall see that the first argument appears to prove that Ursula ought to deduce that Wigner must make/must have made\footnote{In fact the tense ambiguity is core to the perspective on the thought experiment proposed in this essay.} the opposite observation, while the second appears to prove that she should conclude that either observation is possible. Argument 1 proceeds as follows:\\
    \begin{addmargin}[2em]{2em}
	\textit{Suppose that Ursula obtains the outcome corresponding to the state $\ket{-}_{RA}$. Putting ourselves in her shoes, we observe that she can reason as follows: \\
		\begin{addmargin}[2em]{2em}
			After the first three steps of the protocol, the system composed of the two spins and the two agents Alice and Bob is in the state
			\begin{align}
				\sqrt{\frac{2}{3}}\ket{+}_{RA}\ket{0}_S\ket{0}_B+\sqrt{\frac{1}{3}}\ket{1}_{R}\ket{1}_A\ket{1}_S\ket{1}_B.
			\end{align}
			This state is orthogonal to the states $\ket{-}_{RA}\ket{0}_S\ket{0}_B$ and $\ket{-}_{RA}\ket{1}_S\ket{0}_B$, and therefore, since I measured $\ket{-}_{RA}$, it must be the case that Bob measured $\ket{1}_S$. Now I can put myself in Bob's shoes. He can reason as follows:\\
			\begin{addmargin}[2em]{2em}
				The state of S that Alice would have prepared had she measured R to be in the state $\ket{0}_R$ has no overlap with the state $\ket{1}_S$, so had she measured $R$ to be in $\ket{0}_R$, I couldn't have measured $\ket{1}_S$. Therefore, she must have measured R to be in the state $\ket{1}_R$. Now I can put myself in Alice's shoes. She can reason as follows:\\
				\begin{addmargin}[2em]{2em}
					I measured the state $\ket{1}_R$, so I will prepare S in the state $\ket{+}_S$, which Bob will then measure, so that he and S end up in the joint state $\ket{+}_{SB}$. This is an eigenstate of Wigner's measurement, so Wigner will certainly obtain the plus outcome upon making his measurement.\\
				\end{addmargin}
				Alice is reasoning according to the framework of quantum mechanics, so her conclusion must be sound, and Wigner must obtain the plus outcome.\\
			\end{addmargin}
			Bob is reasoning according to the framework of quantum mechanics, so his conclusion must be sound, and Wigner must obtain the plus outcome.\\
		\end{addmargin}
		Ursula is reasoning according to the framework of quantum mechanics, so her conclusion must be sound, and Wigner must obtain the plus outcome if Ursula obtains the minus outcome.}\\
        \end{addmargin}
	Argument 2 is far more direct:\\
    \begin{addmargin}[2em]{2em}
	\textit{
		Applying the standard machinery of quantum to the composite system consisting of the four agents and two qubits, we find that after the protocol is complete, the system is in the state
		\begin{align}
			\Psi&=\frac{\sqrt{3}}{2}\ket{+}_U\ket{+}_{RA}\ket{+}_{SB}\ket{+}_W+\sqrt{\frac{1}{12}}\ket{+}_U\ket{+}_{RA}\ket{-}_{SB}\ket{-}_W\label{eq:final_state}\\
			\nonumber&\hspace{30pt}-\sqrt{\frac{1}{12}}\ket{-}_U\ket{-}_{RA}\ket{+}_{SB}\ket{+}_W+\sqrt{\frac{1}{12}}\ket{-}_U\ket{-}_{RA}\ket{-}_{SB}\ket{-}_W
		\end{align}
		This state has non-zero support on the subspace in which both Ursula and Wigner obtain the outcome $\ket{-}$, and therefore it is possible that if Ursula obtains this outcome, Wigner does as well. }\\
	\end{addmargin}

    The additional classical communication step included in FR's original presentation simply requires Ursula to communicate her prediction of Wigner's outcome to Wigner, so that in the event that both agents measure the outcome $\ket{-}$, Wigner must conclude both that he must have measured $\ket{+}$, by Argument 1, and that he has measured $\ket{-}$, by a fairly natural assumption about what it means to make a measurement. Because by Argument 2 it is in fact possible that both agents do measure $\ket{-}$, this contradiction is in some way ``realized". Clearly, these two arguments produce an absurd conclusion when taken together. Moreover, they both at least appear to be constructed out of valid deductions. 
	
	In my view, insofar as the four agents may validly be called quantum systems, and insofar as the measurement and state preparation operations may be described by unitary operators acting properly on the relevant basis states, the second argument must hold. If it does not, and the wavefunction \eqref{eq:final_state} does not describe the final state of the composite system, then it simply is not the case that that system can be described by standard quantum mechanics. Moreover, if we take Ursula to be reasoning based on Argument 1 (meaning that the unitary operator responsible for updating her brain based on the outcome of her measurement causes evolution consistent with that argument), then there is non-zero support of the final wavefunction on a subspace in which she is \textit{just wrong} about Wigner's observation, insofar as we can identify her mental state of belief with the corresponding wavefunction. Thus, were we to perform this experiment and find that in fact Ursula's reasoning according to Argument 1 is correct (so that whenever she observes the outcome minus, Wigner does not), this would provide evidence that standard quantum mechanics is not valid to describe this system. This would be very odd indeed, as presumably we would have performed separate experiments to validate the use of the unitary operators $M$ and $P_A$ to describe the processes of measurement and state preparation of the individual agents, ruling out the possibility that each of the agents did in fact cause objective collapse upon performing a measurement.
	
	It is all very well to say that Argument 2 must be correct if we assume that the composite system, including the four agents, may be described by quantum mechanics --- indeed that this is \textit{just what it means} to make this claim --- but we are left with the plausibility of Argument 1. Where exactly does it go wrong? Clarifying this point via formalization is the purpose of this essay, and we turn now to this task.
	
	\section{Analysis in terms of incompatible spacelike surfaces}\label{sec:analysis}

	Each of the steps of both Arguments presented above has the same form: an agent $x$, finding himself in state $\psi_x$ at time $t_x$, infers the state $\psi_y$ of another agent $y$ at some time $t_y$. In particular, Argument 1 has the following form: \\
	\begin{addmargin}[2em]{2em}
		\textit{Given that after performing her measurement she is in the state corresponding to the minus outcome, Ursula assigns a state to Bob after his measurement and before Wigner's. Given that he is in this state at this time, Bob assigns a state to Alice after she prepares the system S and before Ursula performs her measurement. Given that she is in this state at this time, Alice assigns a state to Wigner after his measurement.}\\
	\end{addmargin}
	In other words, each step of the deduction corresponds to the determination of the relative state of one agent relative to another, in the sense of \cite{Everett1957}. In fact, it does not seem critical that we actually ascribe these deductions to the agents themselves. We may take a step back from this identification and say ``relative to Ursula, Bob is in the state $\psi$" instead of ``Ursula assigns the state $\psi$ to Bob"\footnote{In a response of \cite{Lazarovici2019} to the Frauchiger-Renner paper -- they argue that the essential error of Argument 1 is that Alice makes an inference based on the assumption that her measurement has collapsed the wavefunction, whereas Ursula's measurement brings the two possible branches back into superposition -- the authors provide the following defense of their approach:\\
		\begin{addmargin}[2em]{2em}
			While the argument of Frauchiger and Renner is concerned with inferences of agents participating in the experiment (and inferences of agents about the inferences of other agents), we make a conscious choice not to take these perspectives but describe the experiment in objective terms. Some readers may worry that this misses the point of the Frauchiger-Renner no-go theorem. But then the point of the Frauchiger-Renner no-go theorem is not a good one to begin with. If different “agents” make inconsistent predictions by applying a quantum theory that makes a consistent prediction, it can only mean that at least one of the agents applies the theory incorrectly.\\
		\end{addmargin}
		I would elaborate on this and say that insofar as the motivation behind the FR though experiment is to probe the possibility of quantum mechanics as a complete theory, and assuming we are materialists about the mind in at least a weak sense, then ``Alice knows/believes/predicts $\phi$ is just another proposition of the same type as $\phi$ itself, from the point of view of an outside observer, and if we are not allowed to take such an external point of view, then it is hard to see how we could talk about a complete theory in the first place.\label{footnote:external_observer}}.

A crucial feature of the argument is that each relative state determination can be associated to a slice through the circuit, and the validity of Argument 1 relies on being able to chain these relative state determinations together. The surfaces used for the various deductive steps are illustrated in Fig. \ref{Fig:causal_structure}, and each determines a universal wavefunction in the Everettian sense, with respect to which relative states may be computed. Note that these surfaces are incomparable with respect to the partial order induced by the operations, and there is no single embedding of the circuit in spacetime that makes all of them surfaces of simultaneity. This point is perhaps somewhat obscured by the careful assignment of specific timestamps to each of the steps of the protocol in the original presentation. 

Now we introduce some formal machinery for analyzing the relevant structure of this type of situation in a streamlined fashion. This will allow the formalization of the two Arguments in a rigorous way, facilitating clarification of the contradiction.

\begin{definition}
A causal structure is a finite directed acyclic graph in which edges are allowed to have both a source and a target, only a source, only a target, or neither.
\end{definition}
In other words, a causal structure is a directed acyclic graph where edges are allowed to ``come in from negative infinity" or ``go off to positive infinity". The edges are meant to represent physical systems or degrees of freedom, and the vertices to represent events involving subsets of these. It is also possible to define a notion of a spacelike surface for a causal structure, which should be understood as motivated by the relativistic notion\footnote{Within special relativity, spacelike surfaces are precisely the possible surfaces of simultaneity, those surfaces for which there is some observer whose reference frame assigns the same time to each point on the surface. In the absence of any particular observer, nothing privileges any one of those surfaces over any other. In the setting of causal structures, spacelike surfaces play the same role. Prior to choosing a particular embedding of the causal structure into spacetime, and an observer in that spacetime, it is meaningless to ask which spacelike surfaces are equal time surfaces, but non-spacelike surfaces can never be equal time surfaces. For example, there is no choice of embedding and observer that puts Alice just prior to her measurement and Wigner just after his measurement at the same time. Even though in their original paper FR do assign times to each of the operations of the protocol -- thus specifying both an embedding and an observer -- their reasoning relies on treating all these surfaces equivalently.}. Because we have stipulated that the graphs under consideration are finite, we can define these inductively:
\begin{definition}
For a fixed causal structure $\mathcal{G}$, let $S_{0}$ be the set of edges with no source. The set of spacelike surfaces of $\mathcal{G}$ is built up inductively by including $S_0$ and declaring that if $S$ is a spacelike surface and for some vertex $v$, $S$ contains all edges for which $v$ is a target, then the set obtained by removing these edges from $S$ and adding all edges for which $v$ is a source is also a spacelike surface. These are the only spacelike surfaces.
\end{definition}
Finally we need a way to associate Hilbert spaces and quantum states to each of these spacelike surfaces in a way consistent with the causal structure and the unitarity of quantum mechanics.
\begin{definition}
Given a causal structure $\mathcal{G}$, a consistent state assignment is a pair $(\mathcal{H},\Psi)$. $\mathcal{H}$ assigns to each edge of $\mathcal{G}$ a finite-dimensional Hilbert space. For each vertex $v$ of $\mathcal{G}$, the product of the dimensions of spaces assigned to the inputs of $v$ equals the product of dimensions of spaces assigned to the outputs. For any spacelike surface $S$ and any subset $\mathcal{A}\subset S$, denote by $\mathcal{H}_\mathcal{A}$ the tensor product of the spaces assigned to the edges in $\mathcal{A}$. Then $\Psi$ assigns to each spacelike surface $S$ a vector $\Psi_S\in\mathcal{H}_S$ such that for any pair of spacelike surfaces $S$ and $S'$, $\Psi_S$ and $\Psi_{S'}$ have identical reduced density operators on $\mathcal{H}_{S\cap S'}$.
\end{definition}
One way to obtain a consistent state assignment once the map $\mathcal{H}$ has been specified is to pick a quantum state assigned to the surface $S_0$, and simply evolve it forward by unitary operators supported on each vertex. In this way, any quantum circuit defines a consistent state assignment. Conversely, any consistent state assignment may be realized in this way\footnote{Note that not all partial consistent state assignments, i.e. choices of consistent states for a subset of spacelike surfaces, may be extended to full consistent state assignments. Consider for example the structure consisting of two distinct vertices, each with a single input and a single output edge. Take the state assigned to the two input edges to be a product state and the state assigned to the two output edges to be entangled. Then any state assigned to the surfaces consisting of the input of one event and the output of the other would have to have different von Neumann entropies for each of its reduced states. This is impossible for a pure state.}.

\subsection{Frauchiger-Renner Argument in terms of relative states}
Equipped with this machinery, we are in a position to reason about the relative state computations along different surfaces involved in the Frauchiger-Renner analysis. We make this notion precise as follows:

\begin{definition}
	Let $\mathcal{G}$ be a causal structure and $(\mathcal{H},\Psi)$ a consistent state assignment. Let $A$ and $B$ be two edges of $\mathcal{G}$, and suppose that there is a spacelike surface $S$ containing $A$ and  $B$. Define the relative state map
	\begin{align}
		\Psi^{A\rightarrow B}:\,\phi\mapsto \mathcal{N}\text{tr}_{S\setminus\set{A,B}}\left[\left(\bra{\phi}_A\otimes 1_{B}\otimes 1_{S\setminus\set{A,B}}\right)\ketbra{\Psi_S}{\Psi_S}\left(\ket{\phi}_A\otimes 1_{B}\otimes 1_{S\setminus\set{A,B}}\right)\right],
	\end{align}
	where $\mathcal{N}$ is a normalization factor.
\end{definition}
We will freely view $\Psi^{A\rightarrow B}$ as a partial function $\mathcal{H}_A\rightarrow\mathcal{H}_B$. The definition does not depend on the choice of spacelike surface containing $A$ and $B$ if there is more than one. Now we can reframe the Frauchiger-Renner protocol as a counterexample to the following claim:
\begin{nontheorem}
	Let $\mathcal{G}$ be a causal structure and $(\mathcal{H},\Psi)$ a consistent state assignment. Let $A_1,  \ldots, A_n$ be edges of $\mathcal{G}$ and suppose that each consecutive pair is contained in some spacelike surface. Suppose also that there is a spacelike surface containing $A_1$ and $A_n$. Then for any $\phi\in \mathcal{H}_{A_1}$, if $\Psi^{A_{n-1}\rightarrow A_n}\circ\cdots\circ \Psi^{A_1\rightarrow A_2}(\phi)$ is defined and pure, it is equal to $\Psi^{A_0\rightarrow A_{n}}(\phi)$.
\end{nontheorem}

\begin{figure}
	\centering
	\includegraphics[width=.7\textwidth]{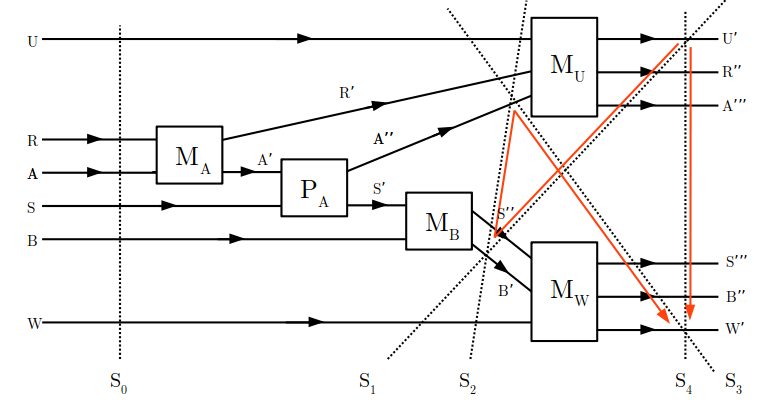}
	\caption{Causal structure of the Frauchiger-Renner protocol with relevant spacelike surfaces indicated. The red arrows pointing along these surfaces indicate the chains of inference used in the two Arguments.}\label{Fig:causal_structure}
\end{figure}

\begin{proof}
The causal structure $\mathcal{G}_{FR}$ is defined to be the structure illustrated in Fig. \ref{Fig:causal_structure}, corresponding to the Frauchiger-Renner protocol. The consistent state assignment $(\mathcal{H}^{FR},\Psi^{FR})$ is defined for the causal structure $\mathcal{G}_{FR}$ by assigning the Hilbert space with basis $\set{\ket{0}, \ket{1}}$ to each edge labeled with an $R$ or $S$, the space with basis $\set{\ket{0}, \ket{1}, \ket{\bot}}$ to each edge labeled with an $A$ or $B$, and the space with basis $\set{\ket{+}, \ket{-}, \ket{\bot}}$ to each edge labeled with a $U$ or $W$ in Fig. \ref{Fig:causal_structure}, and evolving the initial state $\Psi_{S_0}$ defined by the FR protocol forward via the unitary operations corresponding to each event. The states $\Psi_S$ for the spacelike surfaces of interest are given in Table \ref{tab:FRstates}.
\begin{table}[]
    \centering
	\begin{tabular}{lll}
		surface & subsystems & state\\
		$S_0$ & $(U)(R)(A)(S)(B)(W)$ &$\sqrt{\frac{1}{3}}\ket{\bot0\bot0\bot\bot}+\sqrt{\frac{2}{3}}\ket{\bot1\bot0\bot\bot}$\\
		
		$S_1$ & $(U')(R''A''')(S'')(B')(W)$ & $\sqrt{\frac{2}{3}}\ket{++00\bot}+\sqrt{\frac{1}{6}}\ket{++11\bot}-\sqrt{\frac{1}{6}}\ket{--11\bot}$\\
		
		$S_2$ & $(U)(R')(A'')(S'')(B')(W)$ &$\sqrt{\frac{1}{3}}\ket{\bot0000\bot}+\sqrt{\frac{1}{3}}\ket{\bot1100\bot}+\sqrt{\frac{1}{3}}\ket{\bot1111\bot}$\\
		
		$S_3$ & $(U)(R')(A'')(S'''B'')(W')$ & $\sqrt{\frac{1}{6}}\ket{\bot00++}+\sqrt{\frac{1}{6}}\ket{\bot00--}+\sqrt{\frac{2}{3}}\ket{\bot11++}$\\
		
		$S_4$ & $(U')(R''A''')(S'''B'')(W')$& $\sqrt{\frac{3}{4}}\ket{++++}+\sqrt{\frac{1}{12}}\ket{++--}-\sqrt{\frac{1}{12}}\ket{--++}+\sqrt{\frac{1}{12}}\ket{----}$
	\end{tabular}
    \caption{The relevant timelike surfaces $S$ and their states $\Psi_S$ in the consistent state assignment corresponding to the Frauchiger-Renner experiment. Parentheses indicate the grouping of subsystems for the purpose of defining basis states and $\ket{\pm}$ is shorthand for $\frac{1}{\sqrt{2}}\ket{00}\pm\frac{1}{\sqrt{2}}\ket{11}$ for the two-system groups.}
    \label{tab:FRstates}
\end{table}
Writing $\Psi=\Psi^{FR}$, we see by direct computation that
\begin{align}
	\Psi_{S_3}^{A\rightarrow W}\circ\Psi_{S_2}^{B\rightarrow A}\circ\Psi_{S_1}^{U\rightarrow B}\left(\ket{-}\right)=\ket{+}\label{eq:arg1}
\end{align}
whereas
\begin{align}
	\Psi_{S_4}^{U\rightarrow W}\left(\ket{-}\right)=\frac{1}{2}\ketbra{+}{+}+\frac{1}{2}\ketbra{-}{-}.\label{eq:arg2}
\end{align}
Thus the Frauchiger-Renner protocol furnishes a counterexample to the claim.
\end{proof}
Note that Eq. (\ref{eq:arg1}) corresponds to the structure of Argument 1, and Eq. (\ref{eq:arg2}) to the structure of Argument 2. From this point of view, the Frauchiger-Renner protocol simply witnesses the fact that such iterated relative state computations as used in Eq. (\ref{eq:arg1}) are not licensed by quantum mechanics, at least in conjunction with the interpretation that eigenstates of measurement operators may be associated with certain outcomes, which is certainly assumed in the original argumentation.

A point to note is that the mere non-equality of iterated relative states along different paths is not unique to quantum mechanics. A similar non-equality can be obtained in the classical setting with a suitable definition of classical relative states (conditional probability distributions). Consider for example the situation in which Alice flips a fair coin, passing it to Bob if she obtains the outcome ``heads", whereas if she obtains the outcome ``tails", she flips it again before passing it along. If Alice obtains the outcome ``heads" on her first flip, she may reason that Bob will certainly see ``heads". Then she knows that he will assign a probability 2/3 to her having seen ``heads" on the first flip, and a probability 1/3 to her having seen ``tails". But Alice assigns herself a probability of 1 of having seen ``tails". Of course there is nothing odd about this discrepancy, and we do find ourselves frequently in the position of knowing someone else to be uncertain about a fact about which we are ourselves certain. What is strange about the Frauchiger-Renner scenario is that a more circuitous route, passing through more agents, actually seems to yield a certainty not licensed by the direct argument.

Returning to the quantum setting, we may see that if we restrict consideration to a single spacelike surface, argumentation by iterated relative state construction is licensed:
\begin{theorem}
Let $\mathcal{G}$ be a causal structure and $(\mathcal{H},\Psi)$ a consistent state assignment. Let $A_1,  \ldots, A_n$ be edges of $\mathcal{G}$ and suppose that all of these are contained in a single spacelike surface. For some  $\phi^{A_1}\in \mathcal{H}_{A_1}$, suppose $\phi^{A_n}=\Psi^{A_{n-1}\rightarrow A_n}\circ\cdots\circ \Psi^{A_1\rightarrow A_2}(\phi^{A_1})$ is defined and pure. Then $\Psi^{A_1\rightarrow A_{n}}(\phi)=\phi^{A_n}$.
\end{theorem}
\begin{proof}
For $i=2,\ldots, n$, define
\begin{align}
\phi^{A_i}&=\Psi_S^{A_{i-1}\rightarrow A_i}(\phi^{A_{i-1}})
\end{align}
For each $A_i$, choose a basis of $\mathcal{H}_{A_i}$ that includes $\phi^{A_i}$. We can decompose $\Psi_S$ in the product basis (including also a basis for $\mathcal{H}_{S\setminus \set{A_1,A_2,\ldots,A_n}}$). Suppose that $\Psi_S$ has support on a basis state $\ket{\phi^{A_1}\phi_2\phi_3\ldots\phi_n\phi_E}$. Because $\Psi_S^{A_{1}\rightarrow A_2}(\phi^{A_{1}})=\phi^{A_2}$, we must have $\phi_2=\phi^{A_2}$. Proceeding in this manner, we see that we can write
\begin{align}
\Psi_S&=\alpha\ket{\phi^{A_1}}\ket{\phi^{A_2}\phi^{A_3}\ldots\phi^{A_n}\phi_E}+\beta \ket{\phi_\perp^{A_1}}\ket{\Phi'}
\end{align}
and see directly that $\Psi_S^{A_1\rightarrow A_n}=\phi^{A_n}$.
\end{proof}
Thus we locate the source of the Frauchiger-Renner paradox in the use of the iterated relative state construction between multiple spacelike surfaces. Note also that the argument used in this proof, where a state $\Psi_S$ is built up from which the direct relative state can be directly computed, may be reproduced in the classical stochastic setting without restriction to spacelike surfaces. In that case, the quantum state would be replaced by a joint probability distribution over assignments of definite classical states to each system at each time.

\subsection{Quantum logical formulation}
The Frauchiger-Renner contradiction may also be formalized in the language of quantum logic, as initiated in \cite{Birkhoff1936}. This program explores the ways in which quantum mechanical structures can be used to provide semantics to formal languages of propositions, and explores deductive systems that respect this semantics\footnote{A friendly introduction to the notions of syntax and semantics, and related mathematical techniques, can be found in \cite{Winskel1993}.}. In order to make contact with the notion of incompatible spacelike surfaces, we need a syntax that captures the notion of a causal structure, that is, a method for attaching a set of propositions to a causal structure in a way that respects the internal form of these structures.
\begin{definition}
	Let $\mathcal{G}$ be a causal structure. Fix a set $\phi_1,\ldots,\phi_N$ of atomic propositions and sets $\mathcal{A}_1, \ldots, \mathcal{A}_N$ of edges of $\mathcal{G}$ such that for each $i$ there is some spacelike surface $S_i\supset \mathcal{A}_i$. Define the set $\mathcal{L}_\mathcal{G}$ of well-formed propositions, and their supports, inductively as follows:
	\begin{enumerate}
		\item For each $i=1,\ldots, N$, $\phi_i$ is a proposition with support $\mathcal{A}_i$.
		\item If $\phi$ is a proposition with support $\mathcal{A}$, then $\neg\phi$ is also a proposition with support $\mathcal{A}$.
		\item If $\phi$ and $\phi'$ are propositions with supports $\mathcal{A}$ and $\mathcal{A}'$, and there is a spacelike surface $S$ such that $\mathcal{A},\mathcal{A}'\subseteq S$, then $\phi\wedge\phi'$ is a proposition with support $\mathcal{A}\cup\mathcal{A}'$.
	\end{enumerate}\label{def:syntax}
\end{definition} 
Formally speaking, these expressions are purely syntactic --- they are simply formal objects --- but of course they are meant to be read and interpreted (for the moment only informally) in the usual way. A proposition with support $\mathcal{A}$ is a statement about properties localized to the collection of particular systems at particular times contained in $\mathcal{A}$. The operators $\neg$ and $\wedge$ indicate negation (not) and conjunction (and). This is a relatively austere syntax\footnote{For example, we could extend the syntax with logical constants for true and false, localized either to each edge or to each spacelike surface. We could also formally include operators $\vee$ and $\rightarrow$ so that disjunction and implication were included at the formal syntactic level. More interestingly, there is no way within this framework to construct propositions with the informal interpretation ``System S had property $p_1$ at time $t_1$ and property $p_2$ at time $t_2$", as no spacelike surface can contain both of the corresponding edges, except in the trivial case in which S does not participate in any events between $t_1$ and $t_2$. We could imagine various ways of extending the syntax to include propositions of this form. On the other hand, consider the project of describing a complete theory of the universe, including the mental states of all agents. For a proposition like this to be known or believed by an agent, it must in fact correspond to a single-time proposition with support corresponding to some subsystem of that agent's brain (or physical reasoning apparatus of whatever kind). In other words, there must be a \textit{physical record at one time} of the history of system S and its properties (for a discussion of the Frauchiger-Renner thought experiment along these lines, see \cite{Waaijer2021}). Thus it is perhaps reasonable to take this apparent weakness of the formal syntax instead as a strength.}, but will suffice for our purposes, and we will freely use the abbreviation $\phi\longrightarrow\psi$ for the material implication $\neg(\phi\wedge\neg\psi)$.

We now need a way of endowing these propositions with meaning, that is, a semantics. To define a quantum semantics for this syntax, we choose a consistent Hilbert space assignment $\mathcal{H}$ and a consistent state assignment $\Psi$ for the causal structure $\mathcal{G}$.

\begin{definition}
	Let $\mathcal{G}$ be a causal structure and $(\mathcal{H}, \Psi)$ a consistent state assignment. A denotation is an assignment $\denotation{\phi}$ to each $\phi\in \mathcal{L}_\mathcal{G}$ such that
	\begin{enumerate}
		\item $\denotation{\phi_i}$ is a subspace of $\mathcal{H}_{\mathcal{A}_i}$ for $i=1,\ldots, N$
		\item $\denotation{\phi\wedge\phi'}=\denotation{\phi}\cap\denotation{\phi'}$
		\item $\denotation{\neg\phi}=\denotation{\phi}^\perp$
	\end{enumerate}
	The corresponding valuation is the function $\mathcal{V}:\mathcal{L}_\mathcal{G}\rightarrow\set{\texttt{true}, \texttt{false}, \texttt{possible}}$ defined by
	\begin{align}
		\mathcal{V}(\phi)=\begin{cases}
			\texttt{false} & \norm{\denotation{\phi}\Psi_S} =0\\
			\texttt{possible} & 0<\norm{\denotation{\phi}\Psi_S} <1\\
			\texttt{true} &  \norm{\denotation{\phi}\Psi_S} =1
		\end{cases},
	\end{align}
	where $S$ is any spacelike surface containing the support of $\phi$. \label{def:semantics}
\end{definition}
What precisely is meant here by saying that a proposition $\phi$ is ``possible" is of course a contentious matter tied to decades of debate over the Everettian interpretation of quantum mechanics. A straightforward way to interpret the notion of possibility (and indeed the only one that could in principle be probed by experimental realizations of the Frauchiger-Renner protocol) is that it means possibility that an external observer (beyond Alice, Bob, Ursula, and Wigner) would obtain the outcome ``true" upon making a measurement of the system corresponding to the projector $\denotation{\phi}$. But this is not a particularly appealing interpretation of this possibility from the point of view of examining the potential of quantum mechanics to serve as a complete theory --- after all, relative to a complete theory of the universe, there \textit{is} no external observer to make such measurements.  In any case, I don't claim to solve the interpretational problems of quantum mechanics here, so the semantics of the metalanguage employed in Def. \ref{def:semantics} will have to remain unspecified.

Returning to the task of formalization, and to the syntactic side of things, we can introduce a notion of syntactically valid deduction. This will allow the arguments of Frauchiger and Renner to be formalized in a way that demonstrates what precisely goes wrong.
\begin{definition}
	Let $\mathcal{G}$ be a causal structure. A deduction is a finite sequence $\Gamma_0, \ldots,\Gamma_M$ of lists of propositions in $\mathcal{L}_\mathcal{G}$ such that $\Gamma_{n+1}$ is equal to $\Gamma_n$ with a single proposition appended to the end, with that proposition following from one or two of the propositions in $\Gamma_n$ by one of the following deduction rules:
	\begin{enumerate}
		\item $\phi\wedge \phi'\implies\phi$
		\item $\phi\wedge \phi'\implies\phi'$
		\item $\phi, \,\phi'\implies \phi\wedge\phi'$ if the supports of both $\phi$ and $\phi'$ are contained in some spacelike surface $S$
		\item $\phi,\,\phi\longrightarrow\phi'\implies \phi'$
	\end{enumerate}
	The list $\Gamma_0$ is called the set of premises, and the last element of $\Gamma_M$ is called the conclusion. \label{def:deductive}
\end{definition}
Relative to a semantics that assigns binary truth values (true or false) to propositions, a clear requirement for a reliable syntactic deductive system is that if it admits a valid deduction from a set of premises to a conclusion, then whenever all the premises are true, the conclusion is also true --- in other words, that the deductive system is sound. Because we are using the slightly more complex set of truth values that includes ``possible", this requirement should be extended. In particular, if we have a deduction whose premises consist of a set of true propositions and at most a single possible proposition, then the conclusion should be possible or true\footnote{The restriction to at most one merely possible proposition is important given the deductive rules allowed in our system. It is not necessarily the case that the conclusion of a deduction whose premises are possible or true is possible or true. Consider the informally stated deduction: ``$x$ is an orange" and ``$x$ is an apple", therefore ``$x$ is an orange and $x$ is an apple". Clearly there are situations in which we would like to say that both of the premises are possible -- we are reaching into a fruit bowl after coming home on a dark winter night, still wearing our heavy mittens -- but the conclusion is false -- no fruit is both orange and apple. An alternative route would be to include directly at the syntactic level a modal operator corresponding to possibility, and deny the validity of the deduction ``possibly $\phi$" and ``possibly $\phi'$", therefore ``possibly $\phi\wedge\phi'$".}. The deductive system of Def. \ref{def:deductive} can be shown to be unsound in this sense, using the Frauchiger-Renner protocol as a counterexample.
\begin{nontheorem}
	Let $\mathcal{G}$ be a causal structure and $\mathcal{V}$ a valuation on $\mathcal{L}_\mathcal{G}$ defined in terms of some consistent state assignment. Consider a deduction, and suppose that one of the propositions in the premise is possible with respect to $\mathcal{V}$ and the rest are true. Then the conclusion is possible or true with respect to $\mathcal{V}$.
\end{nontheorem}

\begin{proof}
Consider the causal structure and consistent state assignment defined in the proof of the last non-theorem and illustrated in Fig. \ref{Fig:causal_structure} and Table \ref{tab:FRstates}. Define a set of atomic propositions, along with their supports, projectors onto their denotations, and their intended informal interpretations:
\begin{center}
	\def\arraystretch{1.25}
\begin{tabular}{cccp{80mm}}
	Proposition & Support & Denotation & Informal meaning\\
	$\mathcal{U}_-$ & U' & $\ketbra{-}{-}_{U'}$ &Ursula obtains the outcome minus upon measuring Alice and spin R\\
	$\mathcal{B}_1$ & B' & $\ketbra{1}{1}_{B'}$ &Bob obtains the outcome 1 upon measuring spin S\\ 
	$\mathcal{A}_1$ & A'' & $\ketbra{1}{1}_{A''}$&After Alice prepares the spin $S$, she is in the state corresponding to knowing that she measured R to be 1 and prepared S appropriately\\ 
	$\mathcal{W}_-$ & W'& $\ketbra{-}{-}_{W'}$ &Wigner obtains the outcome minus upon measuring Bob and spin S\\
	$\mathcal{W}_+$ & W' & $\ketbra{+}{+}_{W'}$ &Wigner obtains the outcome plus upon measuring Bob and spin S
\end{tabular}
\end{center}
Then we can construct the following deduction:
\begin{center}
\small
	\def\arraystretch{1.25}
	\setlength{\tabcolsep}{0.45em}
\begin{tabular}{ccccccccccc}
$\boxed{\mathcal{U}_-\wedge \mathcal{W}_-}$ & $\mathcal{U}_-\longrightarrow \mathcal{B}_1$ & $\mathcal{B}_1\longrightarrow\mathcal{A}_1$ & $\mathcal{A}_1\longrightarrow\mathcal{W}_+$ \\
$\mathcal{U}_-\wedge \mathcal{W}_-$ & $\boxed{\mathcal{U}_-\longrightarrow \mathcal{B}_1}$ & $\mathcal{B}_1\longrightarrow\mathcal{A}_1$ & $\mathcal{A}_1\longrightarrow\mathcal{W}_+$& $\boxed{\mathcal{U}_-}$ \\
$\mathcal{U}_-\wedge \mathcal{W}_-$ & $\mathcal{U}_-\longrightarrow \mathcal{B}_1$ & $\boxed{\mathcal{B}_1\longrightarrow\mathcal{A}_1}$ & $\mathcal{A}_1\longrightarrow\mathcal{W}_+$& $\mathcal{U}_-$ & $ \boxed{\mathcal{B}_1}$ \\
$\mathcal{U}_-\wedge \mathcal{W}_-$ & $\mathcal{U}_-\longrightarrow \mathcal{B}_1$ & $\mathcal{B}_1\longrightarrow\mathcal{A}_1$ & $\boxed{\mathcal{A}_1\longrightarrow\mathcal{W}_+}$& $\mathcal{U}_-$ & $ \mathcal{B}_1$ & $\boxed{\mathcal{A}_1}$  \\
$\boxed{\mathcal{U}_-\wedge \mathcal{W}_-}$ & $\mathcal{U}_-\longrightarrow \mathcal{B}_1$ & $\mathcal{B}_1\longrightarrow\mathcal{A}_1$ & $\mathcal{A}_1\longrightarrow\mathcal{W}_+$& $\mathcal{U}_-$ & $ \mathcal{B}_1$ & $\mathcal{A}_1$ & $\mathcal{W}_+$ \\
$\mathcal{U}_-\wedge \mathcal{W}_-$ & $\mathcal{U}_-\longrightarrow \mathcal{B}_1$ & $\mathcal{B}_1\longrightarrow\mathcal{A}_1$ & $\mathcal{A}_1\longrightarrow\mathcal{W}_+$&
$\mathcal{U}_-$ & $ \mathcal{B}_1$ & $\mathcal{A}_1$ & $\boxed{\mathcal{W}_+}$ & $\boxed{\mathcal{W}_-}$ \\
$\mathcal{U}_-\wedge \mathcal{W}_-$ & $\mathcal{U}_-\longrightarrow \mathcal{B}_1$ & $\mathcal{B}_1\longrightarrow\mathcal{A}_1$ & $\mathcal{A}_1\longrightarrow\mathcal{W}_+$&
$\mathcal{U}_-$ & $ \mathcal{B}_1$ & $\mathcal{A}_1$ & $\mathcal{W}_+$ & $\mathcal{W}_-$ & $\mathcal{W}_+\wedge\mathcal{W}_-$
\end{tabular}
\end{center}
Boxed propositions in one line are used to deduce the conclusion of the next line. Now it is easy to see that each of the implications in the premise is true (this is exactly what licenses the relative state assignments used in the FR argument and in the previous section), while the first premise is possible (because $\Psi_{S_4}$ has support on $\ket{-}_U\otimes\ket{-}_W$). However, the conclusion $\mathcal{W}_+\wedge\mathcal{W}_-$ is false, as the corresponding subspace $\denotation{\mathcal{W}_-}\cap\denotation{W}_+$ is the trivial subspace $\set{0}\subseteq \mathcal{H}_{W'}$.
\end{proof}
If we restrict the allowed sets of premises to those for which all propositions are supported on subsets of a single spacelike surface, the corresponding theorem would indeed hold. We have the following soundness property:
\begin{theorem}
	Let $\mathcal{G}$ be a causal structure and $\mathcal{V}$ a valuation on $\mathcal{L}_\mathcal{G}$ defined in terms of some consistent state assignment. Consider a deduction, and suppose that one of the propositions in the premise is possible with respect to $\mathcal{V}$ and the rest are true. Moreover, suppose all propositions in the premise are supported on the same spacelike surface. Then the conclusion is possible or true with respect to $\mathcal{V}$.
\end{theorem}
\begin{proof}
Because all propositions in the premise are supported on a single spacelike surface $S$, we may form the conjunction of all of these propositions. As each of the deduction rules preserves this property, we may form the conjunction of all propositions in each step of the deduction. Now Consider a list of propositions $\Gamma'$ derived from another list $\Gamma$ by applying one of the deduction rules. Denoting by $\wedge \Gamma$ the conjunction of all propositions in $\Gamma$, we have $\denotation{\wedge\Gamma'}=\denotation{\wedge\Gamma}$. This holds for rules 1-2 because $\denotation{\phi}\supseteq \denotation{\phi\wedge\phi'}=\denotation{\phi}\cap\denotation{\phi}'$. It holds for rule 3 directly, and for rule 4 we have
\begin{align}
\denotation{\phi\wedge(\phi\longrightarrow\phi')\wedge\phi'}&=\denotation{\phi}\cap\denotation{\neg\left(\phi\wedge\neg\phi'\right)}\cap\denotation{\phi'}=\denotation{\phi}\cap\denotation{\phi\wedge\neg\phi'}^\perp\cap\denotation{\phi'}\\
&=\denotation{\phi}\cap\left(\denotation{\phi}\cap\denotation{\neg\phi'}\right)^\perp\cap\denotation{\phi'}=\denotation{\phi}\cap\left(\denotation{\phi}^\perp\cup\denotation{\phi'}\right)\cap\denotation{\phi'}\\
&=\denotation{\phi}\cap\denotation{\phi'}=\denotation{\phi\wedge \phi'}.
\end{align}
Thus if we take $\Gamma_0$ to be the premise and $\Gamma_n$ the final list of propositions in a deduction, we have $\denotation{\wedge\Gamma_0}=\denotation{\wedge\Gamma_n}$. As $\Gamma_0\subseteq \Gamma_n$, this means that if $\psi$ is the final element of $\Gamma_n$, i.e. the conclusion of the deduction, we have $\denotation{\psi}\supseteq\denotation{\wedge \Gamma_0}$. Thus if the state $\Psi_S$ has support in the subspace $\denotation{\wedge\Gamma_0}$, as it must by assumption that all the propositions in the premise are true except at most one, which is possible, then it has support in the subspace $\denotation{\psi}$. Then $\psi$ is possible or true.
\end{proof}
Taking these two results together, we can identify the use of propositions supported on multiple spacelike surfaces within a single deduction as the source of the Frauchiger-Renner paradox. Again, we can see that an analogous argument holds in the classical stochastic case without restriction to a single spacelike surface, because the natural semantics assigns to the entire causal structure a single sample space with elements corresponding to choices of definite classical state for each system at each time. Note that both this non-theorem and the previous one can be proven using counterexamples of only four qubits, with essentially the same structure as the FR counterexample. I have maintained the systems used in the initial thought experiment in order to facilitate direct comparison with the original argument.

\section{Discussion}
In this essay I have presented two closely related ways of formalizing the Frauchiger-Renner argument in order to pinpoint precisely where it goes wrong. This formalization removes any ambiguity about the source of the contradiction. Perhaps a good characterization of the philosophical error of Frauchiger and Renner is that they try to  have their cake and eat it too with respect to materialism. They claim to analyze a system in which the agents doing measurement, reasoning, and prediction are simply quantum systems like any other, but treat properties like ``Alice believes Wigner will measure plus" in a way different from how they would presumably treat properties like ``the spin S is in the plus state". Thus a sort of underspecified property dualism is smuggled in.

A great deal of prior work has addressed the Frauchiger-Renner argument, some of it touching on similar concerns to those discussed here. \cite{Leegwater2022} analyzes a multi-agent system and establishes that the Born rule cannot hold in all reference frames, but does not explicitly study the chaining of incompatible inferences. \cite{Nurgalieva2019} also take the approach of demonstrating the unsoundness of a quantum semantics for a formal logical structure, introducing the notion of a trust relation between agents and arguing that the non-transitivity of this relation is at the heart of the paradox. \cite{Fraser2023} study an epistemic modal logic and a formalization of the the argument within it, deriving a contradiction. Though related to the discussion here, these analyses do not (at least directly) turn on the notion of combining inferences along different spacelike surfaces, and thus this essay provides a new perspective on the Frauchiger-Renner debate.

\bibliography{refs}

\end{document}